\documentclass[12pt]{article}
\newcommand{\hide}[1]{}

\usepackage{tikz}
\usepackage{etoolbox}
\usepackage{amsmath,amssymb}

\usepackage[utf8]{inputenc}
\usepackage[T1]{fontenc}
\usepackage{flafter}

\newcommand{\cclass}[1]{\textsf{\upshape #1}}

\newcommand{\NP}{\cclass{NP}}

\newcommand{\ac}[1]{\cclass{AC$^{#1}$}}
\newcommand{\circl}{\mathbb{C}}

\newcommand{\function}[2]{\colon #1 \rightarrow #2}

\newcommand{\setdef}[2]{\left\{ \hspace{0.5mm} #1 : \hspace{0.5mm} #2 \right\}}

\newcommand{\Set}[1]{\big\{ #1 \big\}}
\newcommand{\card}[1]{\lvert#1\rvert}

\def\afterthmseparator{.}
\makeatletter
\renewcommand{\@begintheorem}[2]{\trivlist
      \item[\hskip \labelsep{\bf #1\ #2\unskip\afterthmseparator}]\em}
\renewcommand{\@opargbegintheorem}[3]{\trivlist
      \item[\hskip \labelsep{\bf #1\ #2\ (#3)\unskip\afterthmseparator}]\em}
\makeatother
\newtheorem{theorem}{Theorem}[section]
\newtheorem{lemma}[theorem]{Lemma}
\newtheorem{corollary}[theorem]{Corollary}
\newtheorem{remark}[theorem]{Remark}
\newtheorem{proposition}[theorem]{Proposition}

\newcommand{\bull}{\mbox{$\;\;\;$\vrule height .9ex width .8ex depth -.1ex}}
\newenvironment{proof}{\par\smallbreak\noindent{\bf Proof.~}}
{\unskip\nobreak\hfill \bull \par\medbreak}

\newcounter{claim}
\renewcommand{\theclaim}{\Alph{claim}}
\newenvironment{claim}{\refstepcounter{claim}%
\par\medskip\par\noindent{\it Claim~\theclaim.~}~\rm}%
{\par\smallskip\par}
\newenvironment{subproof}{\par\noindent{\sl Proof of Claim~\theclaim.~}}%
{$\,\triangleleft$\par\medskip\par}

\usepackage{enumitem}
\newenvironment{bfenumerate}%
{\mbox{}\begin{enumerate}[label=\arabic*.,font=\upshape\bfseries,nolistsep]}{\end{enumerate}}

\newcounter{oq}

\newcommand{\barG}{\overline{G}}
\newcommand{\barH}{\overline{H}}

\newcommand{\barF}{\overline{F}}
\newcommand{\calH}{\ensuremath{\mathcal{H}}}
\newcommand{\calA}{\ensuremath{\mathcal{A}}}

\newcommand{\calK}{\ensuremath{\mathcal{K}}}
\newcommand{\calL}{\ensuremath{\mathcal{L}}}
\newcommand{\calN}{\ensuremath{\mathcal{N}}}

\newcommand{\calU}{\ensuremath{\mathcal{U}}}
\newcommand{\calW}{\ensuremath{\mathcal{W}}}

\newcommand{\calX}{\ensuremath{\mathcal{X}}}
\newcommand{\calF}{\ensuremath{\mathcal{F}}}

\newcommand{\bbI}[1]{\mathbb{I}(#1)}
\newcommand{\bbZ}{\mathbb{Z}}

\newcommand{\x}{{\mathrm x}}

\usetikzlibrary{calc,arrows,decorations,decorations.pathreplacing,backgrounds,shapes,fit}
\makeatletter
\pgfarrowsdeclare{interval}{interval}
{
  \pgfarrowsleftextend{+-0.5\pgflinewidth}
  \pgfarrowsrightextend{+.5\pgflinewidth}
}
{
  \pgfutil@tempdima=1pt%
  \advance\pgfutil@tempdima by1.25\pgflinewidth%
  \pgfsetdash{}{+0pt}
  \pgfsetrectcap
  \pgfpathmoveto{\pgfqpoint{0pt}{-\pgfutil@tempdima}}
  \pgfpathlineto{\pgfqpoint{0pt}{\pgfutil@tempdima}}
  \pgfusepathqstroke
}

\newif\ifdrawcapoints
\tikzset{%
  ca/.is family,
  ca/.search also={/tikz},
  ca,
  cabasestyle/.style={},
  cabase/.initial=1cm,
  castep/.initial=.2cm,
  caanglestep/.initial=1,
  calevelsx/.initial=0,
  calevelsy/.initial=0,
  calevels/.style={calevelsx=#1,calevelsy=#1},
  at calevels/.style={},
  capoints/.is if=drawcapoints,
  every carc/.style={interval},
  start/.initial={},
  end/.initial={},
  startlevel/.initial=1,
  endlevel/.initial=1,
  level/.style={startlevel=#1,endlevel=#1},
  label/.initial={},
  labelpos/.initial=0.5,
  nodepos/.style={labelpos=#1},% compatibility
  labelangle/.initial={},
  nodeangle/.style={labelangle=#1},
  every label/.style={inner sep=2pt},
  startlabel/.initial={},
  startlabelsep/.initial=\CAcirclesep,
  startlabelanchoroffset/.initial=-90,
  startlabelpos/.is choice,
  startlabelpos/cw/.style={startlabelanchoroffset=+90,startlabelsep=\CAcirclesep,startlabelposauto=false},
  startlabelpos/ccw/.style={startlabelanchoroffset=-90,startlabelsep=\CAcirclesep,startlabelposauto=false},
  startlabelpos/outside/.style={startlabelanchoroffset=+180,startlabelsep=1pt+1.75\pgflinewidth,startlabelposauto=false},
  startlabelpos/inside/.style={startlabelanchoroffset=,startlabelsep=1pt+1.75\pgflinewidth,startlabelposauto=false},
  startlabelposauto/.initial=true,
  every startlabel/.style={inner sep=2pt},
  endlabel/.initial={},
  endlabelsep/.initial=\CAcirclesep,
  endlabelanchoroffset/.initial=+90,
  endlabelpos/.is choice,
  endlabelpos/cw/.style={endlabelanchoroffset=+90,endlabelsep=\CAcirclesep,endlabelposauto=false},
  endlabelpos/ccw/.style={endlabelanchoroffset=-90,endlabelsep=\CAcirclesep,endlabelposauto=false},
  endlabelpos/outside/.style={endlabelanchoroffset=+180,endlabelsep=1pt+1.75\pgflinewidth,endlabelposauto=false},
  endlabelpos/inside/.style={endlabelanchoroffset=,endlabelsep=1pt+1.75\pgflinewidth,endlabelposauto=false},
  endlabelposauto/.initial=true,
  every endlabel/.style={inner sep=2pt},
}
\newenvironment{camodel}[1][]{
  \begin{scope}[ca,#1]
    \path[ca,cabasestyle] (0,0) circle(\pgfkeysvalueof{/tikz/ca/cabase});
    \path[ca,at calevels] (0,0) ellipse
    (\pgfkeysvalueof{/tikz/ca/cabase}+\pgfkeysvalueof{/tikz/ca/calevelsx}*\pgfkeysvalueof{/tikz/ca/castep}
    and \pgfkeysvalueof{/tikz/ca/cabase}+\pgfkeysvalueof{/tikz/ca/calevelsy}*\pgfkeysvalueof{/tikz/ca/castep});
  }{
  \end{scope}
}
\newcommand{\CArc}[1]{%
  \global\def\CAinterval{}%
  \begin{scope}[interval/.append code={\global\def\CAinterval{true}},ca,#1]%
    \pgfkeys{/tikz/ca/cabase/.get=\CAbase}%
    \pgfkeys{/tikz/ca/castep/.get=\CAstep}%
    \pgfkeys{/tikz/ca/start/.get=\CAstart}%
    \pgfkeys{/tikz/ca/end/.get=\CAend}%
    \pgfkeys{/tikz/ca/startlevel/.get=\CAstartlevel}%
    \pgfkeys{/tikz/ca/endlevel/.get=\CAendlevel}%
    \pgfkeys{/tikz/ca/label/.get=\CAlabel}%
    \pgfmathparse{\CAbase+\CAstep*\CAstartlevel}%
    \dimdef\CAstartradius{\pgfmathresult pt}%
    \pgfmathparse{\CAbase+\CAstep*\CAendlevel}%
    \dimdef\CAendradius{\pgfmathresult pt}%
    \pgfmathparse{\CAstart*\pgfkeysvalueof{/tikz/ca/caanglestep}}%
    \let\CAstart\pgfmathresult
    \pgfmathparse{\CAend*\pgfkeysvalueof{/tikz/ca/caanglestep}}%
    \let\CAend\pgfmathresult
    \ifdimgreater{\CAend pt}{\CAstart pt}{%
      \edef\tmpa{\pgfkeysvalueof{/tikz/ca/startlabelposauto}}%
      \ifdefstring{\tmpa}{true}{%
        \tikzset{ca,startlabelpos=cw}%
      }{}%
      \edef\tmpa{\pgfkeysvalueof{/tikz/ca/endlabelposauto}}%
      \ifdefstring{\tmpa}{true}{%
        \tikzset{ca,endlabelpos=ccw}%
      }{}%
    }{}%
    \pgfmathparse{\CAstart\pgfkeysvalueof{/tikz/ca/startlabelanchoroffset}}%
    \let\CAstartlabelanchor\pgfmathresult
    \pgfmathparse{\CAend\pgfkeysvalueof{/tikz/ca/endlabelanchoroffset}}%
    \let\CAendlabelanchor\pgfmathresult
    \edef\CAnode{\pgfkeysvalueof{/tikz/ca/labelangle}}%
    \ifdefempty{\CAnode}{%
      \edef\CAnode{\pgfkeysvalueof{/tikz/ca/labelpos}}%
      \pgfmathparse{\CAstart+(\CAend-\CAstart)*\CAnode}%
      \let\CAnode\pgfmathresult
    }{%
      \edef\tmpa{\pgfkeysvalueof{/tikz/ca/caanglestep}}%
      \pgfmathparse{\CAnode*\tmpa}%
      \let\CAnode\pgfmathresult
    }%
    \ifdimequal{\CAstart pt}{\CAend pt}{%
      \dimdef\CAcirclesep{1pt+1.75\pgflinewidth}%
      \fill (\CAstart:\CAstartradius)
      circle (1pt+1.75\pgflinewidth);
      \ifdefempty{\CAlabel}{}{%
        \ifcsstring{tikz@auto@anchor@direction}{left}{%
          \path
          (canvas polar cs:angle=\CAstart+1,radius=\CAstartradius-1pt-1.75\pgflinewidth) --
          (canvas polar cs:angle=\CAstart-1,radius=\CAstartradius-1pt-1.75\pgflinewidth)
          node[midway,auto,ca,every label] {\CAlabel};
        }{%
          \path
          (canvas polar cs:angle=\CAstart+1,radius=\CAstartradius+1pt+1.75\pgflinewidth) --
          (canvas polar cs:angle=\CAstart-1,radius=\CAstartradius+1pt+1.75\pgflinewidth)
          node[midway,auto,ca,every label] {\CAlabel};
        }%
      }%
    }{%
      \dimdef\CAcirclesep{0pt}%
      \ifdimequal{\CAstartradius}{\CAendradius}{%
        \draw[ca,every carc] (\CAstart:\CAstartradius)
        arc[start angle=\CAstart,end angle=\CAend,radius=\CAstartradius];
      }{%
        \draw[ca,every carc,-] plot[smooth,variable=\t,domain=\CAstart:\CAend]
        ({cos(\t)*(\CAstartradius+(\t-\CAstart)/(\CAend-\CAstart)*(\CAendradius-\CAstartradius))},
        {sin(\t)*(\CAstartradius+(\t-\CAstart)/(\CAend-\CAstart)*(\CAendradius-\CAstartradius))});
        \ifdefempty{\CAinterval}{}{%
          \draw (\CAstart:\CAstartradius-1pt-1.25\pgflinewidth) --
          (\CAstart:\CAstartradius+1pt+1.25\pgflinewidth);
          \draw (\CAend:\CAendradius-1pt-1.25\pgflinewidth) --
          (\CAend:\CAendradius+1pt+1.25\pgflinewidth);
        }
      }%
      \pgfmathparse{\CAstartradius+(\CAnode-\CAstart)/(\CAend-\CAstart)*(\CAendradius-\CAstartradius)}%
      \edef\CAnoderadius{\pgfmathresult pt}%
      \ifdefempty{\CAlabel}{}{%
        \path
        (canvas polar cs:angle=\CAnode+1,radius=\CAnoderadius) --
        (canvas polar cs:angle=\CAnode-1,radius=\CAnoderadius)
        node[midway,auto,ca,every label] {\CAlabel};
      }%
      \ifdrawcapoints
        \fill (canvas polar cs:angle=\CAnode,radius=\CAnoderadius) circle (1pt+1.75\pgflinewidth);
      \fi
    }%
    \path (\CAstart:\CAstartradius) node[outer sep/.expanded=\pgfkeysvalueof{/tikz/ca/startlabelsep},anchor/.expanded=\CAstartlabelanchor,ca,every startlabel] {\pgfkeysvalueof{/tikz/ca/startlabel}};
    \path (\CAend:\CAendradius) node[outer sep/.expanded=\pgfkeysvalueof{/tikz/ca/endlabelsep},anchor/.expanded=\CAendlabelanchor,ca,every endlabel] {\pgfkeysvalueof{/tikz/ca/endlabel}};
  \end{scope}%  
}
\newcommand{\carc}[5][]{%
  \CArc{start={#2},end={#3},level={#4},label={#5},#1}%
}
\makeatother
\tikzset{interval/.style={interval-interval,shorten >=-.5\pgflinewidth,shorten <=-.5\pgflinewidth}}

\usepackage{hyperref}
\usepackage[figure]{hypcap}

\title{Solving the Canonical Representation\protect\linebreak and Star System Problems for\protect\linebreak Proper Circular-Arc Graphs in Logspace}
\author{Johannes Köbler\qquad Sebastian Kuhnert\thanks{Supported by DFG grant  KO 1053/7--1.}\qquad
Oleg Verbitsky\thanks{%
Supported by DFG grant VE 652/1--1.
This work was initiated under support by the Alexander von Humboldt Fellowship.
On leave from the Institute for Applied Problems of Mechanics and Mathematics,
Lviv, Ukraine.}\\[2.5mm]
\normalsize
Humboldt-Universität zu Berlin,
Institut für Informatik\\
\normalsize
Unter den Linden 6,
10099 Berlin, Germany}

\date{}

\begin{document} 

\maketitle

\begin{abstract}
We present a logspace algorithm that constructs a canonical intersection
model
for a given proper circular-arc graph, where \emph{canonical} means that
isomorphic graphs receive identical models. This implies that the recognition
and the isomorphism problems for these graphs are solvable in logspace.
For the broader class of concave-round graphs, which still possess 
(not necessarily proper) circular-arc models, we show that a canonical circular-arc model
can also be constructed in logspace.
As a building block for these results, we design a logspace algorithm
for computing canonical
circular-arc models of circular-arc hypergraphs;
this important class of hypergraphs corresponds to
matrices with the \emph{circular ones property}.

Furthermore, we consider the Star System Problem that
consists in reconstructing a graph from its closed neighborhood hypergraph.
We show that this problem is solvable in logarithmic space
for the classes of proper circular-arc, concave-round, and co-convex graphs.
\end{abstract}

\section{Introduction}

With a family of sets~$\calH$ we associate the \emph{intersection graph}~$\bbI\calH$ on vertex
set~$\calH$ where two sets $A,B\in\calH$ are adjacent if and only if
they have a non-empty intersection. 
We call~$\calH$ an \emph{intersection model} of a graph~$G$
if $G$ is isomorphic to~$\bbI\calH$.
Any isomorphism from~$G$ to~$\bbI\calH$ is called a \emph{representation} of~$G$
by an intersection model.
If $\calH$ consists of 
intervals (resp.\ arcs of a circle),
it is also referred to as an \emph{interval model} (resp.\ an \emph{arc model}).
An intersection model~$\calH$ is \emph{proper} 
if the sets in~$\calH$ are pairwise incomparable by inclusion.
$G$ is called a \emph{(proper) interval graph} if it has a (proper) interval model.
The classes of \emph{circular-arc} and \emph{proper circular-arc} graphs are defined similarly.
Throughout the paper we will use the shorthands \emph{CA} and \emph{PCA}, respectively.

We design a logspace algorithm that for a given PCA graph computes 
a canonical representation by a proper arc model, where \emph{canonical} means that
isomorphic graphs receive identical models.
Note that this algorithm provides a simultaneous solution in logspace
of both the recognition and the isomorphism problems
for the class of PCA graphs.

In~\cite{KoeblerKLV11}, along with Bastian Laubner we gave a logspace
solution for the canonical representation problem of proper interval
graphs.  Though PCA graphs may at first glance appear close relatives
of proper interval graphs, the extension of the result
of~\cite{KoeblerKLV11} achieved here is far from being
straightforward. Differences between the two classes of graphs are
well known and have led to different algorithmic approaches 
also in the past; e.g.\ in~\cite{DengHH96,KaplanN09,LinSS08}.   
One difference, very
important in our context, lies in the relationship of these graph classes to
interval and circular-arc hypergraphs that we will explain shortly.

An \emph{interval hypergraph} is a hypergraph isomorphic to a system of intervals of integers.
A \emph{circular-arc (CA) hypergraph} is defined similarly if, instead of integer intervals,
we consider arcs in a discrete cycle.
With any graph~$G$, we associate its \emph{closed neighborhood hypergraph}
$\calN[G]=\{N[v]\}_{v\in V(G)}$ on the vertex set of~$G$, where for each vertex~$v$
we have the hyperedge~$N[v]$ consisting of~$v$ and all vertices adjacent to~$v$.
Roberts~\cite{Roberts71} 
discovered that $G$~is a proper interval graph if and only if $\calN[G]$ is an
interval hypergraph.
The circular-arc world is more complex.
While $\calN[G]$ is a CA hypergraph whenever $G$ is a PCA graph,
the converse is not always true. PCA graphs
are properly contained in the class of those graphs whose neighborhood hypergraphs are CA.
Graphs with this property are called \emph{concave-round}
by Bang-Jensen, Huang, and Yeo~\cite{Bang-JHY00}
and \emph{Tucker graphs} by Chen~\cite{Chen96}.
The latter name is justified by Tucker's result~\cite{Tucker71}
saying that all these graphs are CA (although not necessarily proper CA).
Hence, it is natural to
consider the problem of constructing arc representations for concave-round graphs.
We solve this problem in logspace and also in a canonical way.

Our working tool is a logspace algorithm for computing a canonical representation of
CA hypergraphs.
This algorithm can also be used to test
in logspace whether a given Boolean matrix has the
\emph{circular ones property}, that is, whether the columns can
be permuted so that the 1-entries in each row form
a segment up to a cyclic shift.
Note that a matrix has this property if and only if it is
the incidence matrix of a CA hypergraph. 
The recognition problem of the circular ones property arises in computational biology, namely in
analysis of circular genomes~\cite{GPZ08,OBS11}.

Our techniques are also applicable to
the \emph{Star System Problem} where, for a given hypergraph~$\calH$,
we have to find a graph~$G$ such that $\calH=\calN[G]$, if such a graph exists.
In the restriction of the problem to a class of graphs~$\mathsf{C}$,
we seek for~$G$ only in~$\mathsf{C}$.
We give logspace algorithms solving the Star System Problem
for PCA and for concave-round graphs.

\break

\subparagraph*{Comparison with previous work.}\mbox{}

\smallskip

\emph{Recognition, model construction, and isomorphism testing.}\quad
The recognition problem for PCA graphs, along with model
construction, was solved in linear time by Deng, Hell, and Huang~\cite{DengHH96}
and by Kaplan and Nussbaum~\cite{KaplanN09};
and in \ac2 by Chen~\cite{Chen97}.
Note that linear-time and logspace results are in general incomparable,
while the existence of a logspace algorithm for a problem implies that it
is solvable in \ac1.
The isomorphism problem for PCA graphs was solved
in linear time by Lin, Soulignac, and Szwarcfiter~\cite{LinSS08}.
In a recent paper~\cite{CurtisLMNSSS13}, Curtis et~al.\ extend 
this result to concave-round graphs.

The isomorphism problem for concave-round graphs
was solved in \ac2 by Chen~\cite{Chen96}.
Circular-arc models for concave-round graphs were known to be constructible
also in \ac2 (Chen~\cite{Chen93}).

Extending these complexity upper bounds to the class of all CA graphs
remains a challenging problem.
While this class can be
recognized in linear time by McConnell's algorithm~\cite{McC03} 
(along with constructing an intersection model),
no polynomial-time isomorphism test for CA graphs is currently known
(see the discussion in~\cite{CurtisLMNSSS13}, where a counterexample
to the correctness of Hsu's algorithm~\cite{Hsu95} is given).
This provides further evidence that CA graphs are algorithmically
harder than interval graphs. For the latter class we have linear-time algorithms for both recognition
and isomorphism due to the seminal work by Booth and Lueker~\cite{BoothL76,LuekerB79},
and a canonical representation algorithm taking logarithmic space is designed in~\cite{KoeblerKLV11}.

The aforementioned circular ones property and the related
\emph{consecutive ones property} (where no cyclic shift is allowed) 
were studied in~\cite{BoothL76,Hsu02,HM03}, where linear-time
algorithms are given;
parallel \ac2\ algorithms were suggested~in~\mbox{\cite{ChenY91,AnnexsteinS98}}.

\medskip

\emph{Star System Problem.}\quad
The decision version of the Star System Problem is in general \NP-complete (Lalonde~\cite{Lalonde81}).
It stays \NP-complete if restricted to non-co-bipartite graphs
(Aigner and Triesch~\cite{AignerT93})
or to $H$-free graphs for~$H$ being a cycle or a path on
at least 5 vertices (Fomin et al.~\cite{FominKLMT11}).
The restriction to co-bipartite graphs has the same complexity as the general graph isomorphism
problem~\cite{AignerT93}.
Polynomial-time algorithms are known for $H$-free graphs for~$H$ being a cycle or a path on
at most 4 vertices~\cite{FominKLMT11} and
for bipartite graphs (Boros et al.~\cite{BorosGZ08}).
An analysis of the algorithms in~\cite{FominKLMT11} for $C_3$- and $C_4$-free
graphs shows that the Star System Problem for these classes is solvable even in logspace,
and the same holds true for the class of bipartite graphs; see~\cite{KKV12}.
Moreover, the problem is solvable in logspace for any logspace-recognizable
class of $C_4$-free graphs, in particular, for 
chordal, interval, and proper interval graphs; see~\cite{KKV12}.

\medskip

A preliminary version of this paper appeared in~\cite{fsttcs}.

\section{Basic definitions}\label{s:preliminaries}

 The vertex set of a graph~$G$ is denoted by~$V(G)$. 
 The \emph{complement of a graph~$G$} is the graph~$\barG$ with
 $V(\barG)=V(G)$
 such that two vertices are adjacent in~$\barG$
 if and only if they are not adjacent in~$G$. 
The set of all vertices at distance at most (resp.\ exactly)~1 from a
 vertex $v\in V(G)$ is called the \emph{closed
   \emph{(resp. \emph{open})} neighborhood} of~$v$ and denoted
 by~$N[v]$ (resp.~$N(v)$).  
Note that $N[v]=N(v)\cup\{v\}$.
We call vertices~$u$ and~$v$ \emph{twins}
 if $N[u]=N[v]$ and \emph{fraternal vertices} if $N(u)=N(v)$.  A
 vertex~$u$ is \emph{universal} if $N[u]=V(G)$. 

The \emph{canonical labeling problem} for a class of graphs~$\mathsf{C}$ is, 
given a graph \mbox{$G\in \mathsf{C}$} with $n$~vertices, to compute a map
$\lambda_G\function{V(G)}{\{1,\ldots,n\}}$ so that the graph~$\lambda_G(G)$, 
the image of~$G$ under~$\lambda_G$
on the vertex set $\{1,\ldots,n\}$, is the same for isomorphic
input graphs.
We say that $\lambda_G$ is
a \emph{canonical labeling} and that $\lambda_G(G)$ is a \emph{canonical form} of~$G$.

Recall that a \emph{hypergraph} is a pair
$(X,\calH)$, where $X$ is a set of vertices and 
$\calH$ is a family of subsets of~$X$, called \emph{hyperedges}.
We will use the same notation~$\calH$ to denote a hypergraph and its hyperedge set
and, similarly to graphs, we will write~$V(\calH)$ referring to the vertex set~$X$
of the hypergraph~$\calH$. 
We will allow \emph{multiple hyperedges}; in this case an isomorphism has to respect multiplicities.

The \emph{complement of a hypergraph~$\calH$} is the hypergraph
$\overline{\calH}=\{\overline{H}\}_{H\in\calH}$ on the same vertex
set, where $\overline{H}=V(\calH)\setminus H$.  Each hyperedge
$\overline{H}$ of~$\overline{\calH}$ inherits the multiplicity of~$H$
in~$\calH$.  With a graph~$G$ we associate two hypergraphs defined on
the vertex set~$V(G)$. The \emph{closed} (resp.\ \emph{open})
\emph{neighborhood hypergraph} of~$G$ is defined by
$\calN[G]=\{N[v]\}_{v\in V(G)}$ (resp.\ by $\calN(G)=\{N(v)\}_{v\in
  V(G)}$).  \emph{Twins in a hypergraph} are two vertices such that
every hyperedge contains either both or none of them. Note that two
vertices are twins in~$\calN[G]$ if
and only if they are twins in~$G$.

Let $X=\{x_1,\dots,x_n\}$. Saying that the sequence $x_1,\dots,x_n$
is \emph{circularly ordered}, we mean that $X$ is endowed with
the (circular successor) relation~$\prec$ under which
$x_i\prec x_{i+1}$ for $i<n$ and $x_n\prec x_1$.
Such a relation~$\prec$ will be referred to as a \emph{circular order} on $X$.
In particular, we will use~$\circl_n$ to denote the initial segment 
of $n$ positive integers with the circular order $1\prec 2\prec \ldots\prec  n\prec 1$.
Note that a circularly ordered set $(X,\prec)$ can be viewed as
a directed cycle.
An ordered pair of elements $a^-,a^+\in X$ determines an \emph{arc}~$A=[a^-,a^+]$
that consists of the points
appearing in the directed path from~$a^-$ to~$a^+$. The elements $a^-$ and~$a^+$
will be referred to as \emph{extreme points} of $A$. This terminology
will be used under the assumption that $A\ne X$, when the extreme points
are uniquely determined by the set $A$. In addition,
the sets $A=\emptyset$ and $A=X$ will be called the \emph{empty arc}
and the \emph{complete arc}, respectively.
A hypergraph~$\calH$ with $V(\calH)=X$ is called 
an \emph{arc system} if all of its hyperedges are arcs.
In this case, the relation $\prec$ will be called a \emph{CA order} of~$\calH$.

An \emph{arc representation of a hypergraph~$\calH$} on $n$ vertices is an isomorphism~$\rho$
from~$\calH$ to an arc system~$\calA$ on~$\circl_n$.
The arc system~$\calA$ is referred to as an \emph{arc model} of~$\calH$.
Hypergraphs having arc representations are called \emph{circular-arc (CA) hypergraphs}.
Note that $\calH$ is a CA hypergraph exactly when it admits a CA order~$\prec$.
Indeed, if $\rho\function{V(\calH)}{\{1,\ldots,n\}}$ is an arc representation of~$\calH$,
we can define $\prec$ by $\rho^{-1}(1)\prec\rho^{-1}(2)\prec\ldots\prec\rho^{-1}(n)\prec\rho^{-1}(1)$.
Conversely, if $v_1\prec v_2\prec\ldots\prec v_n\prec v_1$ is a CA order of~$\calH$,
then $\rho(v_i)=i$ is an arc representation of~$\calH$.

An arc system $\calA$ is \emph{tight} if any two arcs $A=[a^-,a^+]$ and $B=[b^-,b^+]$ in~$\calA$
have the following property: if $A\subseteq B$, then
$a^-=b^-$ or $a^+=b^+$
(note that this condition applies neither to empty nor to complete
arcs that can be in~$\calA$).
A CA order
of~$\calH$ is \emph{tight}, if it makes~$\calH$ a tight arc system.
Furthermore, we call a CA hypergraph \emph{tight} if it admits a tight CA order 
or, equivalently, a tight arc model.
Recognition of tight  CA hypergraphs reduces to recognition of
CA hypergraphs. To see this, given a hypergraph~$\calH$,
define its \emph{tightened hypergraph}~$\calH^\Subset$ by
$
 \calH^\Subset =\calH\cup\setdef{A\setminus B}{A,B\in\calH}.  
$
Then $\calH$ is a tight CA hypergraph if and only if $\calH^\Subset$
is a CA hypergraph (for if $A,B\in\calH$ and $\emptyset\ne B\subseteq A$, 
then $B$~cannot be an inner part of $A$ in any arc model of~$\calH^\Subset$).

The notions of an \emph{interval representation}, an \emph{interval model},
and an \emph{interval order} of a hypergraph 
are introduced similarly to the above, where \emph{interval} means an interval of consecutive integers
within $\{1,\ldots,n\}$. 
Hypergraphs having interval representations are called \emph{interval hypergraphs}.
Since any interval representation is an arc representation, they form a subclass
of CA hypergraphs.

Given a circular order~$\prec$ of a set~$X$, consider the set of
all arcs~$A\subset X$ w.r.t.~$\prec$ excepting the empty arc $\emptyset$ and 
the complete arc $X$. The relation~$\prec$ induces a (lexicographic) circular
order~$\prec^*$ on this set, where $A\prec^*B$ if 
$a^-=b^-$ and $a^+\prec b^+$
or if
$a^-\prec b^-$, $\card{A}=n-1$, and $\card{B}=1$.
The last two conditions say that $A$ is the longest among all arcs with start point $a^-$
and $B$ is the shortest among all arcs with start point~$b^-$.
Let $\calH$ be an arc system such that $\emptyset,V(\calH)\notin\calH$.
By ``restricting'' $\prec^*$ to the hyperedge set $\calH$ we obtain a circular
order~$\prec_\calH$ on~$\calH$: For $A,B\in\calH$ we define
$A\prec_\calH B$ if either $A\prec^*B$ or there exist arcs~$X_1,\dots,X_k\notin\calH$
such that $A\prec^*X_1\prec^*\ldots\prec^*X_k\prec^*B$. We
say that the circular order~$\prec_\calH$ on~$\calH$
is \emph{lifted from} the circular order~$\prec$ on~$V(\calH)$.

An \emph{arc representation of a graph}~$G$ is an
isomorphism~$\alpha\colon V(G)\to\calA$ from~$G$ to the intersection
graph~$\bbI\calA$ of an arc
system~$\calA$ on~$\circl_n$. If~$\emptyset,V(\calA)\notin\calA$ (this always holds
when $G$~has neither an isolated nor a universal vertex), we use the
lifted circular order~$\prec_\calA$ on~$\calA$ to define a circular
order~$\prec_\alpha$ on~$V(G)$, where $u\prec_\alpha v$ if and only if
\mbox{$\alpha(v)\prec_\calA\alpha(u)$.} We call~$\prec_\alpha$ the
\emph{geometric order} on~$V(G)$ associated with~$\alpha$.

\subparagraph*{Roadmap.} %

In Section~\ref{s:CAhgs} we show how to compute a canonical arc representation
for CA hypergraphs in logspace. This procedure will serve as a building
block for our algorithms on PCA and concave-round graphs.
The connections of these classes of graphs to CA hypergraphs are outlined in
Section~\ref{s:linking}. In particular, we make use of the fact that 
the neighborhood hypergraph~$\calN[G]$ of a non-co-bipartite PCA graph~$G$
admits a unique CA
order, which coincides with the geometric order~$\prec_\alpha$
for any proper arc representation~$\alpha$ of~$G$. Based on this, in
Section~\ref{s:graphs} we compute canonical
representations of non-co-bipartite PCA graphs in logspace. To achieve
the same for co-bipartite PCA graphs~$G$ (and all concave-round
graphs), we use the fact that~$\calN(\barG)$ is in this case
an interval hypergraph and 
show how to convert an interval representation
of~$\calN(\barG)$ into an arc representation of~$G$. Finally, in
Section~\ref{s:SSP} we apply the techniques of
Sections~\ref{s:CAhgs}~and~\ref{s:linking} to the Star System Problem.

\section{Canonical arc representations of hypergraphs}\label{s:CAhgs}

In the \emph{canonical representation problem} for CA hypergraphs we
have, for each input hypergraph, to compute its arc representation
such that the resulting arc models are always equal for isomorphic input
hypergraphs.

\begin{theorem}\label{thm:CAhgs}
The canonical representation problem for CA hypergraphs
is solvable in logspace.
\end{theorem}

\begin{proof}
We prove this result by a logspace reduction to the canonical representation problem
for edge-colored interval hypergraphs, which is already known to be in logspace~\cite{KoeblerKLV11}.
Given a hyperedge~$H$ of a hypergraph~$\calH$, we use notation $\barH=V(\calH)\setminus H$.
Let $\calH$ be an input CA hypergraph with $n$~vertices. For each vertex $x\in V(\calH)$
we construct the hypergraph $\calH_x=\{H_x\}_{H\in\calH}$
on the same vertex set, where $H_x=H$ if $x\notin H$ and $H_x=\barH$ otherwise.
Observe that every~$\calH_x$ is an interval hypergraph; cf.~\cite[Theorem~1]{Tucker71}.
Canonizing each~$\calH_x$ using the algorithm from~\cite{KoeblerKLV11},
we obtain $n$~interval representations \mbox{$\rho_x\function{V(\calH)}{\{1,\ldots,n\}}$};
recall that $V(\calH_x)=V(\calH)$. 
Each $\rho_x$ gives us an arc model $\rho_x(\calH)$ of $\calH$, which
is obtained from the corresponding canonical interval model $\rho_x(\calH_x)$ of $\calH_x$
by complementing the intervals corresponding to complemented hyperedges.
Among these $n$~candidates, we choose the lexicographically least arc model
as canonical and output the corresponding arc representation~$\rho_x$.

There is a subtle point in this procedure: We need to distinguish between
complemented and non-complemented hyperedges when canonizing~$\calH_x$;
otherwise reversing the complementation could lead to non-equal models for
isomorphic CA hypergraphs. For this reason we endow each interval hypergraph~$\calH_x$ with
the edge-coloring $c_x\function{\calH_x}{\Set{0,1,2}}$, where $c_x(H_x)=1$ if
$x\in H$ and $c_x(H_x)=0$ otherwise. If both $H$~and~$\barH$ are present
in~$\calH$, this results in two identical (multi)hyperedges that have different colors;
formally, this hyperedge $H_x=(\barH)_x$ receives a special color~$c_x(H_x)=2$.
\end{proof}

\begin{remark}\rm
In the proof of Theorem \ref{thm:CAhgs} we use the canonical representation algorithm
for edge-colored interval hypergraphs designed in~\cite{KoeblerKLV11}. In fact,
in~\cite{KoeblerKLV11} we consider hypergraphs with multiple hyperedges. Nevertheless,
this captures the case of edge-colored hypergraphs because the colors of hyperedges
can be encoded by integers and regarded as multiplicities.

Note also that Theorem \ref{thm:CAhgs} easily extends to \emph{edge-colored} CA hypergraphs.
This requires just a minor modification of the algorithm: When an input hypergraph~$\calH$ is endowed with
an edge-coloring \mbox{$h\function\calH\bbZ$}, we have to endow
the hypergraphs~$\calH_x$ with edge-coloring $h_x(H_x)=3h(H)+c_x(H_x)$,
where $c_x$~is as in the proof of the theorem.
\end{remark}

Translated into the language of matrices, Theorem \ref{thm:CAhgs}
has algorithmic consequences for testing the circular ones property
that was defined in the introduction.

\begin{corollary}
There is a logspace algorithm that decides whether a given Boolean
matrix has the circular ones property and computes an appropriate
permutation of the columns.
\end{corollary}

The \emph{canonical labeling problem} for a class of hypergraphs~$\mathsf{C}$
is defined exactly as for graphs. 
 Notice a similarity between the pairs of notions 
\emph{canonical labeling/canonical form} and
\emph{canonical representation/canonical model} for CA hypergraphs.
The canonical representation algorithm given by Theorem~\ref{thm:CAhgs} also
solves the canonical labeling problem for CA hypergraphs in
logarithmic space. We conclude this section with noting that it can
also be used to compute a canonical labeling for the duals of CA
hypergraphs; this will be needed in Section~\ref{s:SSP}.

Given a hypergraph~$\calH$ and a vertex $v\in V(\calH)$, let
$v^*=\setdef{H\in\calH}{v\in H}$.  The hypergraph
$\calH^*=\setdef{v^*}{v\in V(\calH)}$ on the vertex set
$V(\calH^*)=\calH$ is called the \emph{dual hypergraph} of~$\calH$
(multiple hyperedges in $\calH$ become twin vertices in
$\calH^*$). The map $\varphi\colon v\mapsto v^*$ is an isomorphism
from~$\calH$ to~$(\calH^*)^*$. If~$\calH^*$ is a CA hypergraph, this map can be combined with a
canonical labeling $\lambda$ of~$\calH^*$ in order to obtain a
canonical labeling $\hat\lambda$ of~$\calH$. More precisely,
$\hat\lambda$~is obtained from the map $\lambda'(v)=\setdef{\lambda(H)}{v\in H}$
by sorting and renaming the values of~$\lambda'$.

\begin{corollary}\label{cor:dual}\mbox{}
The canonical labeling problem for hypergraphs whose duals are CA can be solved in logspace.
\end{corollary}

\section{Linking PCA graphs and tight CA hypergraphs}\label{s:linking}

Bang-Jensen et al.~\cite{Bang-JHY00} call a graph~$G$
\emph{concave-round} (resp.\ \emph{convex-round}) if $\calN[G]$
(resp.~$\calN(G)$) is a CA hypergraph.  Since
$\overline{\calN[G]}=\calN(\barG)$, concave-round and convex-round
graphs are co-classes.  Using this terminology, a result of
Tucker~\cite{Tucker71} says that PCA graphs are concave-round, and
concave-round graphs are CA.

To connect the canonical representation problem for PCA and
concave-round graphs to that of CA hypergraphs, we use the
fact that the graph classes under consideration can be characterized
in terms of neighborhood hypergraphs. For concave-round graphs, this
directly follows from their definition, and
we can find accompanying hypergraphs also for PCA graphs.

\begin{theorem}\label{prop:pca_tight_hg}  
A graph $G$ is PCA if and only if $\calN[G]$ is a tight CA
  hypergraph.
\end{theorem}

The forward direction of Theorem~\ref{prop:pca_tight_hg} follows from
Lemma~\ref{lem:geomistight} below. To prove the other direction, we
distinguish two cases. 
If $\barG$ is not
bipartite, then a result of Tucker~\cite{Tucker71} says that $G$ is a
PCA graph whenever $\calN[G]$ is a CA hypergraph.
The case of bipartite $\barG$ is treated in Section~\ref{s:graphs}
where we show that any tight arc
model for $\calN[G]$ can in this case be transformed into a proper arc model for
$G$. Thus, the proof of Theorem~\ref{prop:pca_tight_hg} 
will be completed in Section~\ref{s:graphs}; note that
we will use this result only later in Section~\ref{s:SSP}.

\begin{lemma}\label{lem:geomistight}
  The geometric order~$\prec_\alpha$ on~$V(G)$ associated with a proper
  arc representation~$\alpha$ of a graph~$G$ is a tight CA order
  for the hypergraph $\calN[G]$.
\end{lemma}
\begin{proof}
  Let~$G$ be a PCA graph and let~$\alpha\colon V(G)\to\calA$ be a
  proper arc representation of~$G$. We first show that the
  neighborhood~$N[u]$ of any vertex $u\in V(G)$ is an arc
  w.r.t.\ to the order~$\prec_\alpha$. If $u$ is universal, the claim is
  trivial. Otherwise, let $\alpha(u)=[a^-,a^+]$. We split~$N(u)$ in
  two parts, namely $N^-(u)=\setdef{v\in N(u)}{a^-\in\alpha(v)}$ and
  $N^+(u)=\setdef{v\in N(u)}{a^+\in\alpha(v)}$.
  Indeed, no vertex $v$ is contained in both $N^-(u)$ and
  $N^+(u)$. Otherwise, since $\calA$ is proper, the arcs
  $\alpha(v)$~and~$\alpha(u)$ would cover the whole cycle, both
  intersecting any other arc~$\alpha(w)$, contradicting the assumption
  that $u$ is non-universal.

  Now let $v\in N^+(u)$ and assume that $u\prec_\alpha
  v_1\prec_\alpha\ldots\prec_\alpha v_k\prec_\alpha v$. 
  We claim that every vertex~$v_i$ is in~$N^+(u)$. Indeed, by
  the definition of~$\prec_\alpha$, we have $\alpha(u)\prec_{\calA}
  \alpha(v_1)\prec_{\calA}\ldots\prec_{\calA}\alpha(v_k)\prec_{\calA}\alpha(v)$.
  If $\alpha(v)=[c^-,c^+]$ and $\alpha(v_i)=[b^-,b^+]$, we see that
  $b^-\in(a^-,c^-)$, $b^+\in(a^+,c^+)$ and, hence,
  $a^+\in[b^-,b^+]$.  It follows that $N^+(u)\cup\{u\}$ is
  an arc starting at~$u$. By a symmetric argument, $N^-(u)\cup\{u\}$
  is an arc ending at~$u$. Hence, also~$N[u]$ is an arc, implying
  that~$\prec_\alpha$ is a CA order for $\calN[G]$.

  It remains to show that the CA order~$\prec_\alpha$ is tight.
  Suppose that $N[u]=[u^-,u^+]\subseteq N[v]=[v^-,v^+]$ and $v$ is
  non-universal with $\alpha(v)=[c^-,c^+]$. Let's first assume that
  $u\in N^+(v)=(v,v^+]$. Since $u,v^+\in N^+(v)$, it follows that
  $c^+\in\alpha(u)\cap\alpha(v^+)$. Hence, $u$~and~$v^+$ are adjacent
  or equal, which implies that $u^+=v^+$. If $u\in[v^-,v)$, a symmetric argument
  shows that $u^-=v^-$.
\end{proof}

Theorem~\ref{prop:pca_tight_hg} suggests that, given a tight CA order
of~$\calN[G]$, we can use it to construct a proper arc model for $G$.
For this we need the converse of Lemma~\ref{lem:geomistight}. In
the case that $\barG$ is not bipartite, the following lemma implies
that indeed each CA order of~$\calN[G]$ is the geometric order of some
proper arc representation of~$G$. 

\begin{proposition}\label{prop:any-unique}
  If $G$ is a connected twin-free PCA graph and $\barG$ is not
  bipartite, then $\calN[G]$ has a unique CA order up to reversing.
\end{proposition}

Proposition~\ref{prop:any-unique} can be derived from a result of
Deng, Hell, and Huang~\cite[Corollary~2.9]{DengHH96}.
An alternative, self-contained proof is given in~\cite[Theorem 3.7.1]{rig-conn}.

We close this section by giving a characterization of concave-round
graphs~$G$ with bipartite complement using properties
of~$\calN(\barG)$.  Given a bipartite graph~$H$ and a bipartition
$V(H)=U\cup W$ of its vertices into two independent sets,
by~$\calN_U(H)$ we denote the hypergraph $\Set{N(w)}_{w\in W}$ on the
vertex set~$U$. Note that $\calH_U(H)$ and $\calN_W(H)$ are dual hypergraphs, i.e.,
$(\calN_U(H))^*\cong\calN_W(H)$.
A bipartite graph~$H$ is called \emph{convex}
if its vertex set admits splitting into two independent sets $U$~and~$W$,
such that $\calN_U(H)$ is an interval hypergraph.
If both $\calN_U(H)$~and~$\calN_W(H)$ are interval hypergraphs, $H$~is called \emph{biconvex}~\cite{Spinrad03}.
As $G$ is co-bipartite concave-round if and only if
its complement $H=\barG$ is bipartite convex-round, 
the following fact gives the desired characterization.

\begin{proposition}[Theorem 2.2 in \cite{Tucker74}]\label{prop:biconv}
  A graph~$H$ is bipartite convex-round if and only if it is biconvex 
and if and only if $\calN(H)$ is an interval hypergraph.
\end{proposition}

\begin{figure}[h]
  \centering\footnotesize
  \begin{tikzpicture}[yscale=.57,sibling distance=4.5cm,
    gclass/.style={inner sep=1.5pt,text height=height("d"),text depth=depth("p")}]
    \newcommand{\gclass}[1]{\textrm{\upshape #1}}
    \node[gclass] (CA) {\gclass{CA}}
    child[level distance=1cm] {
      node[gclass] (TCA) {\gclass{concave-round}}
      child[sibling distance=3.5cm,level distance=1.5cm] {
        node[gclass] (PCA) {\gclass{PCA}}
        child[sibling distance=0cm,level distance=1.2cm] {
          node[gclass] (noncobipPCA) {
            \begin{tabular}{c}\\[-1mm]
              \gclass{non-co-bipartite PCA}\\[-1mm]=\gclass{non-co-bipartite concave-round}
            \end{tabular}
          }
        }
      }
      child[sibling distance=4.5cm,level distance=1.5cm] {
        node[gclass] (cobipTCA) {\gclass{co-bipartite concave-round}}
        child[sibling distance=0pt,level distance=1.5cm] {
          node[gclass] (cobipPCA) {\gclass{co-bipartite PCA}}
        }
        child[sibling distance=6.5cm,level distance=1.5cm] {
          node[gclass] {$\gclass{concave-round}\setminus\gclass{PCA}$}
        }
      }
    }
    child[level distance=1cm] {
      node[gclass] (coconvex) {\gclass{co-convex}}
    };
    \path (cobipTCA) ++(3.25cm,0cm) node[gclass] (cobiconvex) {\gclass{co-biconvex}};
    \path (coconvex) edge (cobipTCA);
    \path (PCA) edge (cobipPCA);
    \path (cobipTCA) -- (cobiconvex) node[midway] {$=$};
  \end{tikzpicture}
  \caption{Inclusion structure of the classes of graphs under consideration.}\label{fig:graph-classes}
\end{figure}
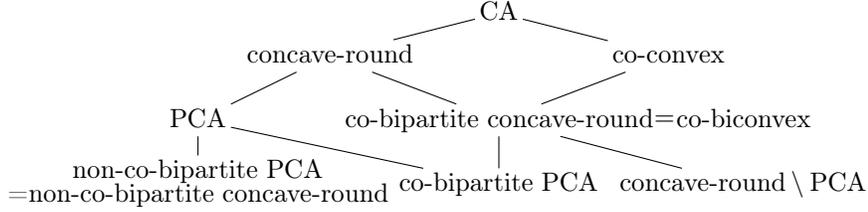

\section{Canonical arc representations of  concave-round and PCA graphs}\label{s:graphs}

We are now ready to present our canonical representation algorithm for concave-round and PCA graphs.
For a given graph, we have to compute its arc representation
such that the resulting arc models are equal for isomorphic input graphs.

\begin{theorem}\label{thm:maingraphs}
  There is a logspace algorithm that solves the canonical arc representation problem 
  for the class of concave-round graphs. Moreover, this algorithm outputs a proper arc
  representation whenever the input graph is PCA.
\end{theorem}

For any class of intersection graphs, a canonical representation
algorithm readily implies a canonical labeling algorithm of the same
complexity.  Vice versa, a canonical representation algorithm readily
follows from a canonical labeling algorithm \emph{and} a
representation algorithm (not necessarily a canonical one). Proving
Theorem~\ref{thm:maingraphs} according to this scheme, we split our
task in two parts: We first compute a canonical labeling~$\lambda$ of
the input graph~$G$ and then we compute an arc representation~$\alpha$
of the canonical form~$\lambda(G)$. Then the
composition~$\alpha\circ\lambda$ is a canonical arc representation
of~$G$. As twins can be easily re-inserted in a (proper) arc
representation, it suffices to compute~$\alpha$ for the twin-free
version of~$\lambda(G)$, where in each twin-class we only keep one
vertex.

We distinguish two cases depending on whether $\barG$ is bipartite;
see Fig.~\ref{fig:graph-classes} for an overview of the involved graph classes.

\subparagraph*{Non-co-bipartite concave-round graphs.}
As mentioned above, any concave-round graph~$G$ whose complement is not
bipartite is actually a PCA graph~\cite{Tucker71}. Hence, we have to
compute a proper arc representation in this case.

\smallskip
\noindent\emph{Canonical labeling.}\quad
We first transform $G$ into its twin-free version~$G'$, where we only
keep one vertex in each twin-class. Let~$n$ be the number of
vertices in~$G'$. We use the algorithm given by Theorem~\ref{thm:CAhgs} to
compute an arc representation~$\rho'$ of~$\calN[G']$. By
Proposition~\ref{prop:any-unique}, $\calN[G']$ has a CA order which is unique
up to reversing. Hence, in order to determine a canonical labeling
of~$G$, it suffices to consider the $2n$~arc
representations~$\rho_1,\dots,\rho_{2n}$ of~$\calN[G]$ that can be
obtained from~$\rho'$ by cyclic shifts and reversing and by
re-inserting all the removed twins. As a canonical labeling~$\rho_i$
of~$G$, we appoint one of these $2n$~variants that gives the
lexicographically least canonical form~$\rho_i(G)$ of~$G$.

\smallskip
\noindent\emph{Proper arc representation.}\quad
As mentioned above, it suffices to find such a representation for the twin-free
graph~$G'$. The arc representation~$\rho'$ of~$\calN[G']$ that we have already
computed provides us with a CA order~$\prec$ for~$\calN[G']$. By
Lemma~\ref{lem:geomistight} and Proposition~\ref{prop:any-unique}, 
%the CA~order~$\prec$ is tight and 
there is a proper arc representation~$\alpha\colon V(G')\to\calA$
of~$G'$ such that $\prec$~coincides with the associated geometric
order~$\prec_\alpha$. In order to construct~$\alpha$ from~$\prec$, we can assume
that no two arcs~$\alpha(v)=[a_v^-,a_v^+]$ and~$\alpha(u)=[a_u^-,a_u^+]$
in~$\calA$ share an extreme point and that $V(\calA)$ consists of exactly $2n$ points. 
A suitable circular order on $V(\calA)$ is uniquely determined
by the conditions that the start points~$a_v^-$ appear in the circle according to~$\prec$,
the same holds true for the end points~$a_v^+$, and that each end
point~$a_v^+$ lies between the start point~$a_{v^+}^-$ and the following
start point, where $v^+$ is the end point of the arc~$N[v]$ w.r.t.~$\prec$.
Using this characterization, $\alpha$ can easily be computed in logspace.
Note that the extreme points of $N[v]=[v^-,v^+]$ are well defined
because no vertex $v$ can be universal; otherwise the arcs
containing the extreme points of $\alpha(v)$ would correspond
to two cliques covering the whole vertex set~$V(G')$.

%%%% hidden comment
\hide{
In the second approach, we first construct a \emph{tight representation} $\rho$ of $G$
on the cycle $(V(G),\prec)$, that is, a representation producing
a tight arc model. For a vertex $v$, consider the arc $N[v]=[v^-,v^+]$
in this cycle. Set $\rho(v)=[v,v^+]$. 
Note that this is an arc representation of~$G$.
Indeed, suppose that $u$ and $v$ are adjacent. Then
either $u\in [v,v^+]$ or $v\in [u,u^+]$
(otherwise $[v^-,v]$ and $[u^-,u]$ would form a clique cover of $V(G)$
by Claim \ref{cl:CA-order}). 
In either case $\rho(v)\cap\rho(u)\ne\emptyset$.
If $u$ and $v$ are not adjacent, $u\notin [v,v^+]$ and $v\notin [u,u^+]$, which
implies that the arcs $\rho(v)$ and $\rho(u)$ are disjoint.
The arc representation $\rho$ is tight because it is obtained from the arc system
$\calN[G]$ that is tight by Lemma \ref{lem:geomistight} and Proposition \ref{prop:any-unique}.
Finally, any tight arc representation of $G$ can be converted
into a proper arc representation. This has been observed by
Tucker~\cite{Tucker71} and Chen~\cite{Chen97} showed that the conversion
can be implemented in~$\ac{1}$. It is not hard to show that it is even possible
in logspace.
}%%%
%%%% end of the hidden cooment

\subparagraph*{Co-bipartite concave-round graphs.}
By Proposition~\ref{prop:biconv}, co-bipartite conca\-ve-round graphs are
precisely the co-biconvex graphs. In fact, even all co-convex graphs are
circular-arc (this is implicit in~\cite{Tucker71}) and we can
compute a canonical arc representation actually for this larger class of graphs.

\smallskip
\noindent
\emph{Canonical labeling.}\quad A logspace algorithm for canonical labeling of
convex graphs, and hence also co-convex graphs, is designed
in~\cite{KoeblerKLV11}.

\smallskip\noindent
\emph{(Proper) arc representation.}\quad
We first recall Tucker's argument~\cite{Tucker71} showing that, if the
complement of $G$ is a convex graph, then $G$ is CA.  We can assume that
$\barG$ has no fraternal vertices as those would correspond to twins
in~$G$.

Let $V(G)=U\cup W$ be a partition of~$\barG$ into independent sets such that
$\calN_U(\barG)$ is an interval hypergraph.
Let $u_1,\ldots,u_k$
be an interval order on~$U$ for~$\calN_U(\barG)$.
We construct an arc representation~$\alpha$ for~$G$ on the cycle~$\bbZ_{2k+2}$
(see Fig.~\ref{fig:cobip-concave-round} for an example)
by setting $\alpha(u_i)=[i,i+k]$ for each $u_i\in U$
and $\alpha(w)=[j+k+1,i-1]$ for each $w\in W$, where $N_{\barG}(w)=[u_i,u_j]$ and
the subscript~$\barG$ means that the vertex neighborhood is considered in the complement of $G$.
Note that $\alpha(w)=\bbZ_{2k+2}\setminus\bigcup_{u\in
  N_{\barG}(w)}\alpha(u)$.  In the case that $N_{\barG}(w)=\emptyset$, we
set $\alpha(w)=[0,k]$.  By construction, all arcs~$\alpha(u)$ for
$u\in U$ share a point (even two, $k$~and~$k+1$), the same holds true
for all~$\alpha(w)$ for $w\in W$ (they share the point~$0$), and any pair
$\alpha(u)$~and~$\alpha(w)$ is intersecting if and only if $u$~and~$w$
are adjacent in~$G$. Thus, $\alpha$~is indeed an arc representation
for~$G$.

\begin{figure}
  \centering
  \begin{tikzpicture}[baseline=-2.75cm,x=1.5cm,y=.66cm,v/.style={circle,fill,inner sep=1.5pt}]
    \node[inner sep=0pt] at (-.75,.5) {(a)};
    \node[v,label=left:$1$\strut] (0) at (0,0) {};
    \node[v,label=left:$2$\strut] (1) at (0,-1) {};
    \node[v,label=left:$3$\strut] (2) at (0,-2) {};
    \node[v,label=left:$4$\strut] (3) at (0,-3) {};
%    \node[v,label=left:$4$\strut] (4) at (0,-4) {};
    \node[v,label=right:$a$\strut] (a) at (1,0) {} edge (0) edge (1) edge (2);
    \node[v,label=right:$b$\strut] (b) at (1,-1) {} edge (1) edge (2);
    \node[v,label=right:$c$\strut] (c) at (1,-2) {} edge (1) edge (2) edge (3); % edge (4);
    \node[v,label=right:$d$\strut] (d) at (1,-3) {} edge (3); % edge (4);
%    \node[v,label=right:$e$\strut] (e) at (1,-4) {} edge (4);
  \end{tikzpicture}\hfill
  \begin{tikzpicture}[y=.175cm,x=.7cm]
    \node[inner sep=0pt] at (-2.5,11.5) {(b)};
    \foreach \x in {0,...,4} {
      \draw[line width=6pt,line cap=round,color=gray!20!white] (\x,-5.5) to
      (\x,4.5) node[black,anchor=south,inner sep=1pt] {\small$\x$};
      \draw[line width=6pt,line cap=round,color=gray!20!white] (\x+6,-5.5) to
      (\x+6,4.5) node[black,anchor=south,inner sep=1pt] {\numdef\xp{\x+5}\small$\xp$};
    }
    \begin{scope}[every node/.style={right,font=\scriptsize,scale=.85,transform shape,inner sep=1.5pt},inode/.style={at end},gray]
      \fill           (4,0) circle (1.5pt) node[anchor=170] {$N(d)$};
      \draw[interval] (2,1) -- (4,1) node[inode] {$N(c)$};
      \draw[interval] (2,2) -- (3,2) node[inode,anchor=190] {$N(b)$};
      \draw[interval] (1,3) -- (3,3) node[inode,anchor=180] {$N(a)$};
    %  \draw[interval] (0,4) -- (2,4) node[inode,anchor=190] {$N(a)$};
    \end{scope}
    \begin{scope}[xshift=5*.7cm,every node/.style={transparent,left},inode/.style={},gray]
      \fill           (4,0) circle (1.5pt) node {$N(d)$};
      \draw[interval] (2,1) -- (4,1) node[inode] {$N(c)$};
      \draw[interval] (2,2) -- (3,2) node[inode] {$N(b)$};
      \draw[interval] (1,3) -- (3,3) node[inode] {$N(a)$};
    %  \draw[interval] (0,4) -- (2,4) node[inode] {$N(a)$};
    \end{scope}
    \begin{scope}[interval,lnode/.style={font=\scriptsize,inner sep=1pt,scale=.85,transform shape}]
    %  \draw (0,-1) -- (4,-1) node[at start,anchor=15,lnode] {$\alpha(0)$};
      \foreach \x in {1,...,4} {
        \draw (\x,-\x-1) -- (\x+5,-\x-1) node[at start,anchor=15,lnode] {$\alpha(\x)$};
      }
    %  \draw (10,4) -- (8,4) node[at start,anchor=190,lnode] {$\alpha(a)$};
      \draw[rounded corners=.525cm] (9,3) -| (11.25,9) -- (-1.25,9) |- (0,3) node[at end,anchor=185,lnode] {$\alpha(a)$};
      \draw[rounded corners=.7cm] (9,2) -| (11.5,10) -- (-1.5,10) |- (1,2) node[at end,anchor=175,lnode] {$\alpha(b)$};
      \draw[rounded corners=.875cm] (10,1) -| (11.75,11) -- (-1.75,11) |- (1,1) node[at end,anchor=180,lnode] {$\alpha(c)$};
      \draw[rounded corners=1.05cm] (10,0) -| (12,12) -- (-2,12) |- (3,0) node[at end,anchor=180,lnode] {$\alpha(d)$};
    \end{scope}
  \end{tikzpicture}
  \caption{(a)~The complement~$\barG$ of a co-bipartite concave-round graph~$G$
    with the bipartition $U=\{0,1,2,3,4\}$ and $W=\{a,b,c,d,e\}$. (b)~An
    interval order of~$\calN_U(\barG)$ (two copies of which are depicted in
    gray) is used to construct an arc representation~$\alpha$ for~$G$ on the circle $\bbZ_{10}$
    (depicted in black); see the text for
    details.}\label{fig:cobip-concave-round}
\end{figure}
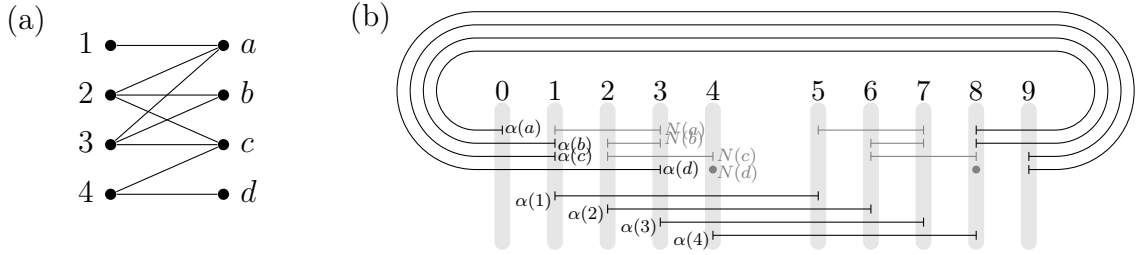

In order to compute~$\alpha$ in logspace, it suffices to compute a
suitable bipartition $\{U,W\}$ of~$\barG$ and an interval order of
the hypergraph~$\calN_U(\barG)$ in logspace.  Finding a bipartition
$\{U,W\}$ such that $\calN_U(\barG)$ is an interval hypergraph can be
done by splitting $\barG$ into connected components $H_1,\dots,H_k$
(using Reingold's algorithm~\cite{Reingold08}) and finding such a
bipartition $\{U_i,W_i\}$ for each component $H_i$. By using the
logspace algorithm of~\cite{KoeblerKLV11} we can actually
compute interval orders of the hypergraphs~$\calN_{U_i}(H_i)$ which
can be easily pasted together to give an interval order
of~$\calN_U(\barG)$.
%%%% hidden comment
\hide{
To this end, we first use the Reingold algorithm \cite{Reingold08} 
to split $H$ into connected components $H_1,\ldots,H_m$, each with bipartition $V(H_i)=U_i\cup W_i$.
Then we run the algorithm of \cite{KoeblerKLV11} on each of the hypergraphs
$\calN_{U_i}(H_i)$ and $\calN_{W_i}(H_i)$, computing interval orders for all those
hypergraphs which are interval. Swapping the notation $U_i$ and $W_i$ whenever needed,
we obtain interval orders for all $\calN_{U_i}(H_i)$ and merge them, obtaining
an interval order of $\calN_U(H)$, where $U=\bigcup_{i=1}^mU_i$.
Thus, we have a logspace algorithm constructing an arc representation of a given 
twin-free co-convex graph. 
}%%% 
%%%% end of the hidden comment
Together with the canonical labeling algorithm this implies that
the canonical arc representation problem for co-convex graphs and, in particular,
for co-bipartite concave-round graphs is solvable in logspace. 

It remains to show that for co-bipartite PCA graphs we can
actually compute a proper arc representation in logspace. 
The existence of such a representation will also complete the
proof of Theorem~\ref{prop:pca_tight_hg} stated in Section~\ref{s:linking}.
As above, we assume
that $G$ is twin-free.  By Lemma~\ref{lem:geomistight}, the
hypergraph~$\calN[G]$ has a tight CA order $\prec$.  We can compute
$\prec$ in logspace by running the algorithm given by
Theorem~\ref{thm:CAhgs} on the tightened
hypergraph~$(\calN[G])^\Subset$.
Any tight CA order of~$\calN[G]$ is also a tight CA order
of~$\calN(\barG)$. Let $V(G)=U\cup W$ be a bipartition of~$\barG$ into
two independent sets.  Note that the restriction of a tight CA
order of~$\calN(\barG)$ to~$\calN_U(\barG)$ is a tight interval
order of the interval hypergraph~$\calN_U(\barG)$.  Retracing
Tucker's construction of an arc representation~$\alpha$ for a
co-convex graph~$G$ (which is outlined above) in the case that the interval order
of~$\calN_U(\barG)$ is tight, we see that~$\alpha$ now gives us a
tight arc model for~$G$.
%%%% hidden comment
\hide{
Indeed, the arc system $\Set{\alpha(u)}_{u\in U}$ is proper
by construction. Furthermore, the arc system $\Set{\alpha(w)}_{w\in W}$
is tight; this follows by construction from the tightness of $\calN_U(H)$.
Finally, any two arcs $\alpha(u)$ for $u\in U$ and $\alpha(w)$ for $w\in W$
are incomparable under inclusion. Indeed, 
$\alpha(w)\not\subset\alpha(u)$ because
$\alpha(w)$ contains $0$ while $\alpha(u)$ does not.
To see that $\alpha(u)\not\subset\alpha(w)$, notice that,
if $N_H(w)=\emptyset$, then
$\alpha(w)$ contains exactly one extreme point of each $\alpha(u)$, and
if $N_H(w)\ni u'$, then 
$\alpha(w)\cap\alpha(u')=\emptyset$ while $\alpha(u')$ contains
one extreme point of each $\alpha(u)$.
}%%%
%%%% end of the hidden comment
Note that, by construction, this model contains no complete arc.
It remains to note that any tight $\alpha$ with this property
can be converted into a proper arc
representation $\alpha'$. Tucker~\cite{Tucker71} described such a
transformation, and Chen~\cite{Chen97} observed
that it can be implemented in~$\ac{1}$. A straightforward inspection shows that it
can even be done in logspace.    
This completes the proof of
Theorem~\ref{thm:maingraphs} and we have additionally proved the
following corollary.

\begin{corollary}\label{thm:coconvex}
  The canonical arc representation problem for co-convex graphs is
  solvable in logspace.
\end{corollary}

\section{Solving the Star System Problem}\label{s:SSP}

In this section, we present logspace algorithms for the 
Star System Problem: Given a hypergraph~$\calH$, find a
graph~$G$ in a specified class of graphs~$\mathsf{C}$ such that
$\calN[G]=\calH$ (if such a graph exists). 
The term \emph{star} refers to the closed neighborhood of a vertex in $G$.
In this terminology, the problem is to identify the center of each star $H$
in the star system $\calH$.
To denote this problem,
we use the abbreviation \emph{SSP}. Note that a logspace
algorithm $\calA$ solving the SSP for a class~$\mathsf{C}$ cannot be
directly used for solving the SSP for a subclass~$\mathsf{C}'$
of~$\mathsf{C}$: If $\calA$ on input~$\calH$ outputs a
solution $G$ in~$\mathsf{C}\setminus\mathsf{C}'$, then we don't know
whether there is another solution $G'$ in~$\mathsf{C}'$. However, if
the SSP for~$\mathsf{C}$ has unique solutions and if membership
in~$\mathsf{C}'$ is decidable in logspace, then it is easy to convert
$\calA$ into a logspace algorithm $\calA'$ solving the~SSP
for~$\mathsf{C}'$.

\begin{theorem}\label{thm:SSP}
  \begin{bfenumerate}
  \item 
The SSP for PCA and for co-convex graphs is solvable in logspace.
\item 
If $G$ is a co-convex graph, then
$N[G]\cong N[G']$ implies $G\cong G'$.
  \end{bfenumerate}
\end{theorem}

The implication stated in Theorem~\ref{thm:SSP}.2 is known to be true
also for concave-round graphs (Chen~\cite{Chen96}).  As a consequence,
since concave-round graphs form a logspace decidable subclass of the
union of PCA and co-convex graphs, we can also solve the SSP for
concave-round graphs in logspace.

The proof of Theorem~\ref{thm:SSP} is given in the rest of this
section. We design logspace algorithms
$\calA_1$ and $\calA_2$ solving the SSP for non-co-bipartite PCA
graphs and for co-convex graphs, respectively. Since by
Theorem~\ref{thm:SSP}.2, the output of $\calA_2$ is unique up to
isomorphism, we can easily combine the two algorithms to obtain a
logspace algorithm $\calA_3$ solving the SSP for all PCA graphs: On
input~$\calH$ run $\calA_1$ and $\calA_2$ and check if one of
the resulting graphs is PCA (recall that co-bipartite PCA graphs are
co-convex; see Fig.~\ref{fig:graph-classes}).

Clearly, it suffices to consider the case that the input
hypergraph~$\calH$ is connected. 

\subparagraph*{Non-co-bipartite PCA graphs.}
Let~$\calH$ be the given input hypergraph and assume that
$\calH=\calN[G]$ for a PCA graph~$G$. By
Theorem~\ref{prop:pca_tight_hg}, $\calH$ has to be a tight CA
hypergraph, a condition that can be checked by testing if the
tightened hypergraph~$\calH^\Subset$ is CA. Since~$G$ is
concave-round, Proposition~\ref{prop:biconv} implies that $G$ is
co-bipartite if and only if $\calN(\barG)=\overline{\calH}$ is an
interval hypergraph.  It follows that the SSP on~$\calH$ can only have
a non-co-bipartite PCA graph as solution if $\calH^\Subset$ is CA and
$\overline{\calH}$ is not interval.  Both conditions can be checked in
logspace using the algorithms given by Theorem~\ref{thm:CAhgs}
and~\cite{KoeblerKLV11}.  Further, it follows by
Theorem~\ref{prop:pca_tight_hg} and Proposition~\ref{prop:biconv} that in
this case any SSP solution for~$\calH$ is a non-co-bipartite
PCA graph (which is also connected because $\calH$ is assumed to be
connected).

By considering the quotient hypergraph with respect to twin-classes, 
we can additionally assume that $\calH$ is twin-free.

In order to reconstruct $G$ from $\calH$, we have to choose the center
in each star $H\in\calH$.
The following lemma considerably restricts this choice.

\begin{lemma}\label{lem:order-iso}
  Let $G$ be a connected, non-co-bipartite and twin-free PCA graph and
  let $\prec$ be a circular order on~$V(G)$ that is a CA order
  of~$\calN[G]$. Then $u\prec v$ holds exactly when
  $N[u]\prec_{\calN[G]} N[v]$, where $\prec_{\calN[G]}$~is the
  circular order on~$\calN[G]$ lifted from~$\prec$.
\end{lemma}

\begin{proof}
First of all, note that the circular order $\prec_{\calN[G]}$
on~$\calN[G]$ is correctly defined because a non-co-bipartite PCA graph 
has no universal vertex (we observed this fact in Section \ref{s:graphs}).
By the same reason we can use the notation
$N[u]=[u^-,u^+]$ w.r.t.~$\prec$. 

\begin{claim}\label{cl:CA-order}
For any vertices $u,v\in V(G)$, the following conditions are met.
  \begin{enumerate}
  \item $u$ divides $N[u]=[u^-,u^+]$ into two parts
     $[u^-,u]$ and $[u,u^+]$ that both are cliques in~$G$.
  \item $v\in [u,u^+]$ if and only if $u\in [v^-,v]$.
  \item If $v\in [u,u^+]$, then $v^-\in [u^-,u]\text{\ and\ }u^+\in [v,v^+]$.
  \item If $v\in [u,u^+]$ and $u\prec v$, then $u^-$, $v^-$,
    $u$, $v$, $u^+$, and $v^+$ occur under the
    order~$\prec$ exactly in this circular sequence, where some
    of the neighboring vertices except $u^-$~and~$v^+$ may coincide.
  \end{enumerate}
\end{claim}

\begin{subproof}
  By Lemma~\ref{lem:geomistight} and Proposition~\ref{prop:any-unique}, 
there is a proper arc representation~$\alpha$
of~$G$ such that $\prec$~coincides with the associated geometric
order~$\prec_\alpha$. Parts 1 and 2 will follow from a simple geometric
observation: $v\in [u,u^+]$ if and only if $\alpha(v)$ contains the end point of $\alpha(u)$.
To see this equivalence, it suffices to notice that $\alpha(v)$ cannot
contain both extreme points of $\alpha(u)$; otherwise $\alpha(v)$ and $\alpha(u)$
would cover the entire circle and, hence, both $v$ and $u$ would be universal.

\smallskip

1. $[u,u^+]$ is a clique because all arcs $\alpha(v)$ for $v\in [u,u^+]$
share the end point of $\alpha(u)$. Similarly, $[u^-,u]$ is a clique because 
all arcs $\alpha(v)$ for $v\in [u^-,u]$ share the start point of~$\alpha(u)$.

\smallskip

2. This part is true because $\alpha(v)$ contains the end point of $\alpha(u)$
if and only if $\alpha(u)$ contains the start point of~$\alpha(v)$.

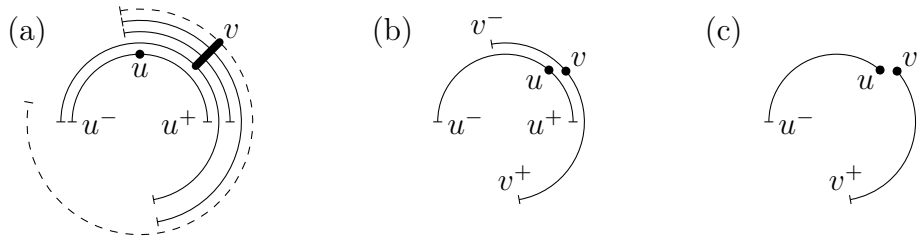
\begin{figure}
  \centering
  \begin{tikzpicture}[baseline=0cm]
    \begin{camodel}[cabase=.75cm,castep=.15cm]
       \node at (-1.5,1.2) {(a)};
      \carc[startlabel=$u^-$,endlabel=$u^+$,startlabelpos=inside,endlabelpos=inside]{180}{0}{1}{}
      \carc[swap,nodeangle=90,every label/.style={inner sep=5pt}]{90}{90}{1}{$u$}
      %\carc[swap,nodeangle=45,every label/.style={inner sep=3pt}]{45}{45}{1}{$v$}
      \draw[line cap=round,line width=3pt] (45:1.05cm) -- (45:1.5cm) node[above right,inner sep=0pt] {$v$};
      \carc{180}{-80}{2}{}
      \carc{100}{0}{3}{}
      \carc{100}{-80}{4}{}
      \carc[dashed]{100}{-190}{5}{}
    \end{camodel}
  \end{tikzpicture}\hfil
    \begin{tikzpicture}[baseline=0cm]
      \begin{camodel}[cabase=.75cm,castep=.15cm]
        \node at (-1.5,1.2) {(b)};
        \carc[startlabel=$u^-$,endlabel=$u^+$,startlabelpos=inside,endlabelpos=inside]{180}{0}{1}{}
        \carc[swap,nodeangle=50,every label/.style={inner sep=3pt}]{50}{50}{1}{$u$}
        \carc[startlabel=$v^-$,endlabel=$v^+$,startlabelpos=outside,endlabelpos=inside]{100}{-80}{2}{}
        \carc[nodeangle=40,every label/.style={inner sep=3pt}]{40}{40}{2}{$v$}
      \end{camodel}
    \end{tikzpicture}\hfil
    \begin{tikzpicture}[baseline=0cm]
      \begin{camodel}[cabase=.75cm,castep=.15cm]
        \node at (-1.5,1.2) {(c)};
        \carc[startlabel=$u^-$,startlabelpos=inside]{180}{50}{1}{}
        \carc[swap,nodeangle=50,every label/.style={inner sep=3pt}]{50}{50}{1}{$u$}
        \carc[endlabel=$v^+$,endlabelpos=inside]{40}{-80}{2}{}
        \carc[nodeangle=40,every label/.style={inner sep=3pt}]{40}{40}{2}{$v$}
      \end{camodel}
    \end{tikzpicture}
  \caption{(a) 
Proof of Claim \protect\ref{cl:CA-order}.3.  The most
    inward arc $[u^-,u^+]$ represents~$N[u]$. The other four arcs show
    possible positions of $[v^-,v^+]=N[v]$, where the outmost, dashed
    variant is actually impossible if $u\prec v$ by Claim
    \protect\ref{cl:CA-order}.4.\quad
(b)-(c) The two cases in the proof of Lemma
      \protect\ref{lem:order-iso}.}\label{fig:abc}
\end{figure}

The remaining parts will be deduced from parts 1 and~2.

\smallskip

3. If the two conditions in~part~2 are true, then $v^-$,
    $u$, and $v$ occur in this circular order. Since $[v^-,v]$ is a
    clique, all vertices in $[v^-,u)$ are adjacent to~$u$ and hence,
    $v^-\in [u^-,u]$. The second containment follows
    by symmetry. All possible mutual
    positions of $N[u]$~and~$N[v]$ are shown in
    Fig.~\ref{fig:abc}.a.
  
4. By parts~2~and~3 it follows that $u^-$, $v^-$, $u$, $v$,
    $u^+$, and $v^+$ occur in this circular order, where
    $v^+$~and~$u^-$ may be swapped or can coincide; see Fig.~\ref{fig:abc}.a. To show that the
    condition $u\prec v$ rules out the last two possibilities,
    assume the contrary and note that then $[v,v^+]$ and $[u^-,u]$
would be two cliques covering the vertex set~$V(G)$.
\end{subproof}

  In order to prove the lemma, it suffices to show that $u\prec v$ implies
  $N[u]=[u^-,u^+]\prec_{\calN[G]} N[v]=[v^-,v^+]$.  To this end we
  show that there is no third vertex~$w$ such that the arcs $N[u]$,
  $N[w]$, and~$N[v]$ appear in this sequence under the circular
  order~$\prec_{\calN[G]}$.

  Suppose first that $u$~and~$v$ are adjacent. Then it follows from
  Claim~\ref{cl:CA-order}.4, that the vertices $u^-$, $v^-$, $u$, $v$,
  $u^+$, and, $v^+$ appear in this circular sequence; see
  Fig.~\ref{fig:abc}(b). We split our analysis into three
  cases, depending on the position of~$w$ on the cycle~$(V(G),\prec)$.
  If $w\in(v,v^+]$, then Claim~\ref{cl:CA-order}.3 implies that
  $w^-\in[v^-,v]$ and $v^+\in[w,w^+]$. If $w^-\ne v^-$, then $N[u]$,
  $N[v]$, and~$N[w]$ appear in this sequence
  under~$\prec_{\calN[G]}$. The same holds true if $w^-=v^-$ because
  then the arc $[w^-,w^+]$ has to be longer than the arc $[v^-,v^+]$
(note that, if also $u^-=v^-$, then $[u^-,u^+]$ is shorter than $[v^-,v^+]$).
  The case that $w\in[u^-,u)$ is similar.  If $w\in(v^+,u^-)$, then
  $w^-\in(v,u^-)$, and again $N[w]$ cannot be intermediate.

  Suppose now that $u$~and~$v$ are not adjacent.  
  It follows that $N[u]=[u^-,u]$ and
  $N[v]=[v,v^+]$; see Fig.~\ref{fig:abc}(c).  
  By Claim~\ref{cl:CA-order}.1, both $N[u]$ and $N[v]$ are cliques.
  Again we have
  to show that for no third vertex~$w$, the arcs $N[u]$, $N[w]$,
  and~$N[v]$ appear in this sequence under~$\prec_{\calN[G]}$. This is
  clear if $w^-\in(v,u^-)$.  This is also so if
  $w^-=v$, because then the arc $[v,v^+]$ must be shorter than the arc $[w^-,w^+]$
by Claim~\ref{cl:CA-order}.3.
  Finally, note that the remaining case $w^-\in[u^-,v)$ is
  not possible. Indeed, in this case $v\notin N[w]$, for else the
  non-adjacent vertices $u$~and~$v$ would belong to the clique
  $[w,w^+]$. Hence, it would follow that
  $N[w]=[w^-,w^+]\subsetneq[u^-,u^+]=N[u]$, contradicting the fact that
  $N[u]$ is a clique.
\end{proof}

Lemma~\ref{lem:order-iso} states that the mapping $v\mapsto N[v]$ is an
isomorphism between the two directed cycles~$(V(G),\prec)$
and~$(\calN[G],\prec_{\calN[G]})$. Since there are exactly $n$
such isomorphisms, we get exactly $n$ candidates
$f_1,\dots,f_n$ for the mapping $v\mapsto N[v]$. Hence, all we have to
do is to use the algorithm given by Theorem~\ref{thm:CAhgs} to compute a CA
order~$\prec$ of~$\calH$ and the corresponding lifted order
$\prec_\calH$ in logspace. Now for each isomorphism $f$
between~$(V(\calH),\prec)$ and~$(\calH,\prec_{\calH})$ we have to check if
selecting $v$ as the center of the star $f(v)$ results in a graph~$G$,
that is, if for all $v,u\in V(\calH)$ it holds that~$v\in f(v)$ and that
$v\in f(u)$ exactly when $u\in f(v)$.

\subparagraph*{Co-convex graphs.}
Let~$\calH$ be the given hypergraph and assume that $\calH=\calN[G]$
for a co-convex graph~$G$. To facilitate the exposition, suppose first
that the bipartite complement~$\barG$ is connected,
with vertex partition $U,W$.  Then $
\overline{\calH}=\calN(\barG)=\calN_U(\barG)\cup\calN_W(\barG)$,
where the vertex-disjoint hypergraphs $\calU=\calN_U(\barG)$ and
$\calW=\calN_W(\barG)$ are dual (i.e., $\calU^*\cong\calW$), both
connected, and at least one of them is interval, say,~$\calU$.
Note also that, since $\barG$ is connected, $\overline{\calH}$
has no isolated vertex, that is, every vertex is contained in some hyperedge.
We need a simple auxiliary fact.

\begin{lemma}\label{lem:Ncomps}
  Let $K$ be a graph without isolated vertices and
  let~$\calL$ be a connected component of $\calN(K)$. Denote
  $U=V(\calL)$.  Then either $U$ is an independent set in~$K$ or $U$
  spans a connected component of~$K$.  Moreover, if $U$ is
  independent, then there is a connected component of~$K$ that is a
  bipartite graph with $U$~being one of its vertex classes.
\end{lemma}

\begin{proof}
  If $U$~is not independent in~$K$, it contains at least two adjacent
vertices~$u_1$ and~$u_2$. Let $K'$ denote the connected component of~$K$
containing $u_1$ and~$u_2$. 
By connectedness of~$\calL$, the set~$U$ contains both neighborhoods~$N_K(u_1)$ and~$N_K(u_2)$. 
We can apply this observation to each edge
along any path in~$K'$. It readily follows that
$V(K')\subseteq U$. 
In fact, $V(K')=U$ because otherwise $\calL$
would be disconnected.

Assume now that $U$~is independent in~$K$. 
Consider a vertex $u\in U$ and a vertex~$w$ adjacent to~$u$
in~$K$. Let $\calL'$ be the connected component of~$\calN(K)$ containing~$w$.
As shown above, the set of vertices $W=V(\calL')$ is independent in~$K$
(otherwise $W$ would contain $u$).
By connectedness of~$\calL$ and~$\calL'$, once we have an edge~$uw$ between~$U$ and~$W$, 
we have $N_K(w)\subseteq U$ and $N_K(u)\subseteq W$.
Let $K'$ denote now the connected component of~$K$ containing $u$ and~$w$.
This observation is applicable to each edge along any path in~$K'$.
It follows that $K'$~is bipartite with one vertex class included in~$U$
and the other in~$W$. In fact, the vertex classes of~$K'$ coincide with~$U$ and~$W$ by connectedness of~$\calL$ and~$\calL'$.
\end{proof}

Denote $\calK=\overline{\calH}$ and assume that $\calK=\calN(K)$ for
some graph~$K$, possibly different from $\barG$.
Since $\calK$ has no isolated vertex, $K$ also has none.
Lemma~\ref{lem:Ncomps} implies that either~$K$
is a connected bipartite graph with partition $U,W$ or~$K$ has two
connected components $K_1$ and $K_2$ with $V(K_1)=U$ and $V(K_2)=W$.
However, the second possibility leads to a contradiction. Indeed,
since the hypergraph $\calN(K_1)=\calU$ is interval,
Proposition~\ref{prop:biconv} implies that $K_1$ is bipartite,
contradicting the connectedness of~$\calU$.  Therefore, $K$ must be
connected and bipartite with vertex partition $U,W$.

Recall that the \emph{incidence graph} of a hypergraph~$\calX$ is the
bipartite graph with vertex classes~$V(\calX)$ and~$\calX$ where $x\in
V(\calX)$ and $X\in\calX$ are adjacent if $x\in X$ (if~$X$ has
multiplicity~$k$ in~$\calX$, it contributes $k$~fraternal vertices in
the incidence graph).  Since $K$ is isomorphic to the incidence graph
of the hypergraph~$\calU$ (as well as~$\calW$), $K$ is
reconstructible from~$\calK$ up to isomorphism and, in particular,
$K\cong\barG$.  Thus, the solution to the SSP on~$\calH$ is unique up
to isomorphism.\footnote{%
The uniqueness result of Boros et al.~\cite{BorosGZ08}
implies a somewhat weaker fact, namely the uniqueness up to isomorphism
\emph{within} the class of co-convex graphs.}

After these considerations we are ready to describe our logspace
algorithm for solving the SSP for the class of co-convex co-connected
graphs.  Given a hypergraph~$\calH$, we first check if
$\overline{\calH}$ has exactly two connected components, say
$\calU$~and~$\calW$. This can be done by running Reingold's
reachability algorithm~\cite{Reingold08} on the
intersection graph~$\bbI{\overline{\calH}}$.  If this is not the case,
there is no solution in the desired class.  Otherwise, we construct
the incidence graph~$F$ of the hypergraph~$\calU$ (or of~$\calW$,
which should give the same result up to isomorphism) and take its
complement~$\barF$.  Note that this works well even if $\barF$ has
twins: the twins in~$V(\calU)$ are explicitly present, while the twins
in~$V(\calW)$ are represented by multiple hyperedges in~$\calU$.

As argued above, if the SSP on~$\calH$ has a co-convex co-connected
solution, then the closed neighborhood hypergraph
$\calF=\calN[\barF]$ of~$\barF$ is isomorphic to~$\calH$. However, it may
not be equal to~$\calH$. In this case we compute an
isomorphism~$\varphi$ from~$\calF$ to~$\calH$ or, the same task,
from~$\overline{\calF}$ to~$\overline{\calH}$.  This can be done by
the algorithms of~\cite{KoeblerKLV11} and Corollary~\ref{cor:dual},
because at least one of the connected components of
$\overline{\calF}\cong\overline{\calH}$ is an interval hypergraph and
the other component is isomorphic to the dual of an interval
hypergraph. Now, the isomorphic image $G=\varphi(\barF)$ of~$\barF$ is
the desired solution to the SSP on~$\calH$ as
$\calN[\varphi(\barF)]=\varphi(\calN[\barF])=\calH$.

  If we do not succeed with establishing an isomorphism between~$\calF$ and~$\calH$,
this implies that there is no solution in the desired class. Alternatively,
we could check from the very beginning whether
one of the hypergraphs $\calU$~and~$\calW$ is interval and $\calU^*\cong\calW$.

Consider now the general case when $\calH=\calN[G]$
for a co-convex graph~$G$ with not necessary connected complement $\barG$.
Note that universal vertices of~$G$ are easy to identify in~$\calH$:
those are the vertices contained in every hyperedge of~$\calH$.
We can remove all such vertices from~$\calH$, solve the SSP for the 
reduced hypergraph, and then restore a solution for~$\calH$.
The last step can be done in a unique way. 
We will, therefore, assume that $G$~has no universal vertex or, equivalently,
$\overline{\calH}=\calN(\barG)$ has no isolated vertex.

If $\barG$ consists of~$k$ connected components $H_1,\ldots,H_k$, where $H_i$~is
a bipartite graph with vertex classes~$U_i$ and~$W_i$, then
$\calK=\overline{\calH}$ consists of~$2k$ connected components
$\calU_i=\calN_{U_i}(H_i)$ and $\calW_i=\calN_{W_i}(H_i)$, each pair being dual.
Moreover, it can be supposed that all $\calU_i$ are interval hypergraphs.

Assume that $\calK=\calN(K)$ for any other graph~$K$.
By Lemma~\ref{lem:Ncomps}, for each connected component $\calL\in\Set{\calU_i,\calW_i}_{i=1}^k$
either $V(\calL)\in\Set{U_i,W_i}_{i=1}^k$ spans a connected component of~$K$ or
there is another connected component $\calL'$ such that $V(\calL)\cup V(\calL')$
spans a connected component of~$K$ that is a bipartite graph. 
Note that in the latter case $\calL$ and~$\calL'$ have to be dual hypergraphs, i.e.,
$\calL'\cong\calL^*$.
Recall that, by Proposition~\ref{prop:biconv}, no $U_i$ can alone span a connected
component of~$K$. It readily follows that $K$ consists of~$k$ connected
bipartite components $K_1,\ldots,K_k$, where vertex classes~$Y_i$ and~$Z_i$ of each~$K_i$
span connected components of~$\calK$. Moreover, we can enumerate $K_1,\ldots,K_k$ so that
the components of~$\calK$ spanned by~$Y_i$ and~$Z_i$ are isomorphic to~$\calU_i$ and~$\calW_i$.
Since both $H_i$ and~$K_i$ are isomorphic to the incidence graph
of the hypergraph~$\calU_i$ (as well as $\calW_i$),
the graphs~$K$ and~$\barG$ are isomorphic and
the solution to the SSP on~$\calH$ is unique up to isomorphism.

This analysis suggests the following logspace algorithm solving
the SSP for the class of co-convex graphs without universal vertices.
Given a hypergraph~$\calH$, we first check if $\overline{\calH}$~has an even number of connected components
that can be split into pairs $\calU_i$ and~$\calW_i$ so that $\calU_i$~is an interval hypergraph
and $\calW_i\cong\calU_i^*$. 
This step can be done by using Reingold's algorithm 
and the algorithm of~\cite{KoeblerKLV11}. A desired solution exists if and only if
this is possible.

Note that some of the hypergraphs~$\calW_i$ can also be interval.
Then the set~$\Set{\calU_i}_{i=1}^k$ can be chosen in essentially different 
(non-isomorphic) ways; however, all these choices will give isomorphic outcomes
(as all choices of~$\Set{\calU_i}_{i=1}^k$ are equivalent up to isomorphism
and taking duals).

Then, for each~$i$, we construct the incidence graph~$F_i$ of the hypergraph~$\calU_i$,
form the graph~$F$ as the vertex-disjoint union of all $F_i$, and take its complement~$\barF$.

By the already established uniqueness,
the closed neighborhood hypergraph $\calF=\calN[\barF]$ is isomorphic to~$\calH$.
We find an isomorphism~$\phi$ from~$\calF$ to~$\calH$ or, the same,
from~$\overline{\calF}$ to~$\overline{\calH}$.
We do it componentwise by running the algorithms of~\cite{KoeblerKLV11} and Corollary~\ref{cor:dual}
on the connected components of~$\overline{\calF}$ and~$\overline{\calH}$.
The isomorphic image $G=\phi(\barF)$ is a solution as
$\calN[\phi(\barF)]=\phi(\calN[\barF])=\calH$.

\section{Conclusion}\label{sec:conclusion}

By Theorem~\ref{thm:maingraphs}, there is a logspace algorithm that solves the
canonical arc representation problem for PCA graphs, where the constructed models
are proper. \emph{Unit CA graphs} are CA graphs that admit a PCA model where all
arcs have equal length. 
%Unit CA graphs are a proper subclass of PCA graphs~\cite{Tucker74}. 
The unit arc representation problem for such
graphs can be solved in linear time~\cite{LinS08,KaplanN09}. Can it also be
solved in logspace? 
The \emph{unit interval} representation problem is solved in logspace in~\cite{KoeblerKLV11}.

In Section~\ref{s:SSP}, we solve the Star System Problem for PCA graphs and
concave-round graphs in logspace. Is this also possible for other classes of
circular-arc graphs? Furthermore, can one extend the result of
Theorem \ref{thm:SSP}.2 about the uniqueness of a solution to this problem?

In analogy to convex graphs, Liang and Blum~\cite{LB95} call a bipartite
graph~$G$ with vertex classes $U$~and~$V$ \emph{circular convex}, if
$\calN_U(G)=\{N_G(u)\}_{u\in U}$ is a CA hypergraph. We remark that our logspace
algorithm for canonical representation of CA hypergraphs can be used to solve
the canonical labeling problem for circular convex graphs in logspace. Indeed,
the approach of~\cite{KoeblerKLV11} to convert a canonical representation
algorithm for interval hypergraphs into a canonical labeling algorithm for
convex graphs can be easily adapted to this setting.

\subparagraph*{Acknowledgement.}
We thank Bastian Laubner for useful discussions at the early stage of this work.

\end{document}

% Local Variables: 
% ispell-local-dictionary: "en"
% End: 